\documentclass[12pt,halfline,a4paper]{ouparticle}
\usepackage{graphicx}
\usepackage{multirow}
\usepackage{array}
\usepackage{amsmath,amssymb,amsfonts}
\usepackage{amsthm}
\usepackage{mathrsfs}
\usepackage[title]{appendix}
\usepackage{xcolor}
\usepackage{textcomp}
\usepackage{manyfoot}
\usepackage{booktabs}
\usepackage{algorithm}
\usepackage{algorithmicx}
\usepackage{algpseudocode}
\usepackage{listings}
\newtheorem{theorem}{Theorem}

\newtheorem{definition}{Definition}
\begin{document}

\title{Topological analysis of entropy measure using regression models for diamond structure}

\author{%
\name{H. M. Nagesh$^{1}$, Girish V. R$^{2}$}
$^{1,2}$\address{Department Science and Humanities, PES University, \\ Bangalore, Karnataka, India.}
\email{nageshhm@pes.edu$^{1}$, girishvr1@pes.edu$^{2}$.}}

\abstract{A numerical parameter, referred to as a topological index, is used to represent the molecular structure of a compound by analyzing its graph-theoretical characteristics. Topological indices are predictive methods for the physicochemical properties of chemical compounds in the context of quantitative structure-activity relationship (QSAR) and quantitative structure-property relationship (QSPR) analysis. Graph entropies have evolved into information-theoretic tools for exploring the structural information of molecular graphs. In this paper, we compute the Nirmala index, as well as the first and second inverse Nirmala index, for the diamond structure using its M-polynomial. Furthermore, entropy measures based on Nirmala indices are derived for the diamond structure. The comparison of the Nirmala indices and corresponding entropy measures is presented through numerical computation and 2D line plots. A regression model is built to investigate the relationship between the Nirmala indices and corresponding entropy measures.}

\date{}


\maketitle
 \maketitle
\section{Introduction}
\label{sec1}
Given a simple, connected, and undirected graph with a non-empty vertex set $V(\Upsilon)$ and an edge set $E(\Upsilon)$, let $\Upsilon =(V(\Upsilon), E(\Upsilon))$ be an ordered pair. The total number of edges incident to a vertex $v \in V(\Upsilon)$ is its \emph{degree} and is denoted as $d_{\Upsilon}(v)$. For the edge $e=uv$, $u$ and $v$ are the end vertices of the edge $e \in E(\Upsilon)$. 
\newpage
The study and modeling of the structural characteristics of chemical compounds is the focus of the area of chemical graph theory. Chemical compounds are depicted here as graphs, with atoms serving as vertices and chemical bonds serving as edges between atoms. Molecular structures are mathematically investigated in this field of study using theoretical, computational, and graphical methods \cite{1}. Topological indices are mathematical parameters that are obtained from the molecular graph of a chemical compound. These indices are used for predicting the physical attributes, chemical composition, and biological activity of the substance. In the context of quantitative structure-property relationship (QSPR) and quantitative structure-activity relationship (QSAR) investigations, their significance is especially apparent \cite{2,3}. 

Given an edge set $E(\Upsilon)$ of a graph $\Upsilon$, the topological index based on degrees \cite{4} is given by
    \begin{equation*}
I(\Upsilon)= \displaystyle \sum_{uv \in E(\Upsilon)} f(d_{\Upsilon}(u), d_{\Upsilon}(v)),
    \end{equation*}
where $f(x, y)$ is a non-negative and symmetric function that depends on the mathematical formulation of the topological index.

Several topological indices have been reported in the literature and have shown benefits in several fields, including drug development, biology, chemistry, computer science, and physics. Proposed by H. Wiener in 1947, the Wiener index was the first and most studied topological index. One noteworthy application of it is the prediction of paraffin boiling temperatures \cite{5}. Another well-known degree-based topological index is the Randi\'c index first introduced by Milan Randi\'c in 1975. Its significance for drug development is widely used \cite{6}. For more information on topological indices and their applications, readers are referred to \cite{7,8}. 

Several attempts have been made to increase the prediction capability by adding new indices to the degree-based topological indices. Kulli in \cite{9} recently introduced a novel degree-based topological index of a molecular graph $\Upsilon$, called the \emph{Nirmala index}.

\begin{equation}
N(\Upsilon)= \displaystyle \sum_{uv \in E(\Upsilon)} \sqrt{d_{\Upsilon}(u)+d_{\Upsilon}(v)}
    \end{equation}

The first inverse Nirmala index $IN_{1}(\Upsilon)$ and second inverse Nirmala index $IN_{2}(\Upsilon)$ of a molecular graph $\Upsilon$ are defined as follows by Kulli \cite{10} later in 2021.

\begin{equation}
IN_{1}(\Upsilon)= \displaystyle \sum_{uv \in E(\Upsilon)} \sqrt{\frac{1}{d_{\Upsilon}(u)}+\frac{1}{d_{\Upsilon}(v)}}=\displaystyle \sum_{uv \in E(\Upsilon)} \left( \frac{1}{d_{\Upsilon}(u)}+\frac{1}{d_{\Upsilon}(v)} \right)^{\frac{1}{2}}
    \end{equation}
 \begin{equation}
IN_{2}(\Upsilon)= \displaystyle \sum_{uv \in E(\Upsilon)} \frac{1}{\sqrt{\frac{1}{d_{\Upsilon}(u)}+\frac{1}{d_{\Upsilon}(v)}}}=\displaystyle \sum_{uv \in E(\Upsilon)} \left( \frac{1}{d_{\Upsilon}(u)}+\frac{1}{d_{\Upsilon}(v)} \right)^{-\frac{1}{2}}
    \end{equation} 

    \newpage
In the past, topological indices were computed using the standard mathematical definition of a number of topological indices. There have been several attempts to look into a compact method that can recover a large number of topological indices for a particular class. In this regard, the concept of a general polynomial was developed, whose values of the necessary topological indices at a given point are produced by its derivatives, integrals, or a combination of both.  In this context, the idea of a general polynomial was found, whose derivatives, integrals, or a mix of both yield the values of the required topological indices at a given point. For example, the Hosoya polynomial \cite{11} recovers the distance-based topological indices, whereas the NM-polynomial \cite{12} generates the neighborhood degree sum-based topological indices. In 2015, Deutsch and Kla\v{z}ar \cite{13} introduced the notion of M-polynomial as a method for determining the degree-based topological index. 

\begin{definition} $(\cite{13})$
\emph{The M-polynomial of a graph $\Upsilon$ is defined as:
\begin{center}
    $M(\Upsilon; x,y)=\displaystyle \sum_{\delta \leq i \leq j \leq \Delta} m_{i,j}(\Upsilon)x^{i}y^{j}$,
\end{center}
where $\delta = min\{d_{\Upsilon}(u) | u \in V(\Upsilon)\}$, $\Delta = max\{d_{\Upsilon}(u) | u \in V(\Upsilon)\}$, and $m_{ij}$ is the number of edges $uv \in E(\Upsilon)$ such that $d_{\Upsilon}(u)=i, d_{\Upsilon}(v)=j \, (i,j \geq 1)$}.
\end{definition}

The M-polynomial-based derivation formulas to compute the different Nirmala indices are listed in Table 1.

\hspace{5mm}
\begin{center}
 \textbf{Table 1}. Relationship between the M-polynomial and Nirmala indices for a graph $\Upsilon$.      
\end{center}

\begin{table}[h!]
\centering
\renewcommand{\arraystretch}{4}
\begin{tabular}{||c| c |c| c ||} 
 \hline
 Sl. No & Topological Index & $f(x,y)$ & Derivation from $M(\Upsilon;x,y)$  
 \\ [0.5ex] 
 \hline\hline
 1  & Nirmala index (N) & $\sqrt(x+y)$ & $D_{x}^{\frac{1}{2}}J(M(\Upsilon; x,y))|_{x=1}$ \\
\hline
2  & First inverse Nirmala index $(IN_1)$ & $\sqrt(\frac{x+y}{xy})$ & $D_{x}^{\frac{1}{2}}JS_{y}^{\frac{1}{2}}S_{x}^{\frac{1}{2}}(M(\Upsilon; x,y))|_{x=1}$\\
\hline
 3  & Second inverse Nirmala index $(IN_2)$ & $\sqrt(\frac{xy}{x+y})$ & $S_{x}^{\frac{1}{2}}JD_{y}^{\frac{1}{2}}D_{x}^{\frac{1}{2}}(M(\Upsilon; x,y))|_{x=1}$  \\
 [1ex] 
 \hline
\end{tabular}
\end{table} 

\newpage
Here, $D_{x}^{\frac{1}{2}}(h(x,y))=\sqrt{x \cdot \frac{\partial(h(x,y))}{\partial x}} \cdot \sqrt{h(x,y)}$;  $D_{y}^{\frac{1}{2}}(h(x,y))=\sqrt{y \cdot \frac{\partial(h(x,y))}{\partial y}} \cdot \sqrt{h(x,y)}$; 
$S_{x}^{\frac{1}{2}}(h(x,y))=\sqrt{\displaystyle \int_{0}^{x} \frac{h(t,y)}{t} }dt \cdot \sqrt{h(x,y)}$;
$S_{y}^{\frac{1}{2}}(h(x,y))=\sqrt{\displaystyle \int_{0}^{y} \frac{h(x,t)}{t} }dt \cdot \sqrt{h(x,y)}$; and $J(h(x,y))=h(x,x)$ are the operators. 

For more information on degree-based topological indices using M-polynomial, we refer the readers to \cite{14,15,16,17,18,19,20}.

Shannon \cite{21} proposed the fundamental concept of entropy in 1948, characterizing it as a measurement of the uncertainty or unpredictability of the information contained in a system, represented by a probability distribution. The structural information of graphs, networks, and chemical structures was then analyzed using entropy. Graph entropies have become more useful in the recent several years in a variety of disciplines, including mathematics, computer science, biology, chemistry, sociology, and ecology. Graph entropy measures can be classified into different types such as intrinsic and extrinsic measures, and they correspond to the probability distribution with graph invariants (edges, vertices, etc.). 
For further information on degree-based graph entropy measures and their uses, readers are referred to \cite{22,23,24,25}.

\subsection{Entropy of a graph in terms of edge-weight}
 In 2014, Chen et al. \cite{26} proposed the concept of the entropy of edge-weighted graphs as follows. 
 
 Consider the edge-weight graph $\Upsilon = (V(\Upsilon), E(\Upsilon), \omega(uv))$, in which $E(\Upsilon)$ represents the set of edges, $V(\Upsilon)$ represents the set of vertices, and $\omega(uv)$ represents the weight of an edge $uv \in E(\Upsilon)$. Then, given an edge weight, the entropy of a graph is defined as follows: 
 
 \vspace{3mm} 
\begin{align*}
ENT_{\omega}(\Upsilon)=&-\displaystyle \sum_{u^{'}v^{'} \in E(\Upsilon)} \frac{\omega(u^{'}v^{'})}{\displaystyle \sum_{uv \in E(\Upsilon) } \omega(uv) } 
log \left(\frac{\omega(u^{'}v^{'})}{\displaystyle \sum_{uv \in E(\Upsilon) } \omega(uv) } \right) \\
= & -\displaystyle \sum_{u^{'}v^{'} \in E(\Upsilon)} \frac{\omega(u^{'}v^{'})}{\displaystyle \sum_{uv \in E(\Upsilon) } \omega(uv) } 
\left[log(\omega(u^{'}v^{'})) -log \left(\displaystyle \sum_{uv \in E(\Upsilon)} \omega(uv) \right) \right] \\
=& log \left(\displaystyle \sum_{uv \in E(\Upsilon)} \omega(uv) \right) -\displaystyle \sum_{u^{'}v^{'} \in E(\Upsilon)} \frac{\omega(u^{'}v^{'})}{\displaystyle \sum_{uv \in E(\Upsilon) } \omega(uv) } 
log(\omega(u^{'}v^{'})) 
\end{align*}
Hence,
\begin{align}
ENT_{\omega}(\Upsilon)=& log \left(\displaystyle \sum_{uv \in E(\Upsilon)} \omega(uv) \right) - \frac{1}{\left(\displaystyle \sum_{uv \in E(\Upsilon) } \omega(uv) \right)}              \displaystyle \sum_{u^{'}v^{'} \in E(\Upsilon)} \omega(u^{'}v^{'})  log(\omega(u^{'}v^{'}))
\end{align} 

Recently in \cite{27}, Virendra Kumar et al. introduced the concept of the Nirmala indices-based entropy. To study this, they examined the relevant information function $\omega$ as a function related to the equations (1-3) that define the Nirmala indices.
 \\\\
\textbf{Nirmala entropy}: Let $\omega(uv)=\sqrt{d_{\Upsilon}(u)+d_{\Upsilon}(v)}$. Then from the definition of the Nirmala index as given in Equation (1), we have
\begin{center}
$ \displaystyle \sum_{uv \in E(\Upsilon)}\omega(uv)=\displaystyle \sum_{uv \in E(\Upsilon)} \sqrt{d_{\Upsilon}(u)+d_{\Upsilon}(v)}=N(\Upsilon)$  
\end{center}
Hence, using equation (4), the Nirmala entropy of a graph $\Upsilon$ is given by 
\begin{equation}
 ENT_{N}(\Upsilon)=log(N(\Upsilon))  - \frac{1}{N(\Upsilon)}              \displaystyle \sum_{uv \in E(\Upsilon)} \sqrt{d_{\Upsilon}(u)+d_{\Upsilon}(v)} \times log (\sqrt{d_{\Upsilon}(u)+d_{\Upsilon}(v)})
\end{equation} \\\\
\textbf{First inverse Nirmala entropy}: Let $\omega(uv)=\sqrt{\frac{1}{d_{\Upsilon}(u)}+\frac{1}{d_{\Upsilon}(v)}}$. Then from the definition of the first inverse Nirmala index as given in Equation (2), we have,
\begin{center}
$ \displaystyle \sum_{uv \in E(\Upsilon)}\omega(uv)=\displaystyle \sum_{uv \in E(\Upsilon)} \sqrt{\frac{1}{d_{\Upsilon}(u)}+\frac{1}{d_{\Upsilon}(v)}}=IN_{1}(\Upsilon)$  
\end{center}
Hence, using equation (4), the first inverse Nirmala
entropy of a graph $\Upsilon$ is given by 
\begin{equation}
 ENT_{IN_{1}}(\Upsilon)=log(IN_{1}(\Upsilon))  - \frac{1}{IN_{1}(\Upsilon)}              \displaystyle \sum_{uv \in E(\Upsilon)} \sqrt{\frac{1}{d_{\Upsilon}(u)}+\frac{1}{d_{\Upsilon}(v)}} \times log \left( \sqrt{\frac{1}{d_{\Upsilon}(u)}+\frac{1}{d_{\Upsilon}(v)}} \right) 
\end{equation}
\textbf{Second inverse Nirmala entropy}: Let $\omega(uv)=\frac{\sqrt{d_{\Upsilon}(u) \cdot  d_{\Upsilon}(v)} }{\sqrt{d_{\Upsilon}(u) + d_{\Upsilon}(v)} }$.
Then from the definition of the second inverse Nirmala index as given in Equation (3), we have,
\begin{center}
$ \displaystyle \sum_{uv \in E(\Upsilon)}\omega(uv)=\displaystyle \sum_{uv \in E(\Upsilon)} \frac{\sqrt{d_{\Upsilon}(u) \cdot  d_{\Upsilon}(v)} }{\sqrt{d_{\Upsilon}(u) + d_{\Upsilon}(v)} }=IN_{2}(\Upsilon)$  
\end{center}
Hence, using equation (4), the second inverse Nirmala
entropy of a graph $\Upsilon$ is given by 
\begin{equation}
 ENT_{IN_{2}}(\Upsilon)=log(IN_{2}(\Upsilon))  - \frac{1}{IN_{2}(\Upsilon)}              \displaystyle \sum_{uv \in E(\Upsilon)} \frac{\sqrt{d_{\Upsilon}(u) \cdot  d_{\Upsilon}(v)} }{\sqrt{d_{\Upsilon}(u) + d_{\Upsilon}(v)} } \times log \left( \frac{\sqrt{d_{\Upsilon}(u) \cdot  d_{\Upsilon}(v)} }{\sqrt{d_{\Upsilon}(u) + d_{\Upsilon}(v)} } \right) 
 \end{equation}
 
 \subsection{Methodology}
In this paper, the tetrahedral structure of diamond is used. Further, we use Shannon's entropy measures by employing novel information functions derived from various Nirmala index definitions. We conducted a comprehensive mathematical and computational exploration of these measures within the tetrahedral structure of a diamond. The rest of the paper is organized as follows: Section 2 discusses the crystallographic structure of diamonds. In Section 3, we compute the Nirmala indices of diamond structure using its M-polynomial, which allows us to calculate the Nirmala indices-based entropy measures. Section 4 compares the Nirmala indices and their associated entropy measures through numerical data and a 2D line plot. Section 5 deals with a regression model to illustrate how the estimated Nirmala indices and associated entropy values fit the curve. Lastly, a discussion and a conclusion are presented in sections 5 and 6, respectively.  
 
\section{Crystallographic structure of diamond}
Diamond, a metastable form of carbon, exhibits a covalent bonding pattern where each carbon atom forms bonds with four surrounding carbon atoms. These atoms are arranged in a diamond lattice, a variation of the face-centered cubic crystal structure. Because of their extraordinary hardness, diamonds have many important industrial uses. The diamond's extraordinarily high refractive power gives it remarkable brightness. A well-cut diamond will reflect more light into the observer's eye than a gem with lower refractive power, giving it a more brilliant appearance. The high dispersion of diamonds originates from the dispersion of white light into the color of the spectrum as it passes through the stone.

Figure 1 depicts the sequential formation of the diamond's tetrahedral structure and Figure 2 depicts the tetrahedral structure of a diamond \cite{28}.
\hspace{5mm}
\begin{figure}[h!]
\centering
\includegraphics[width=130mm]{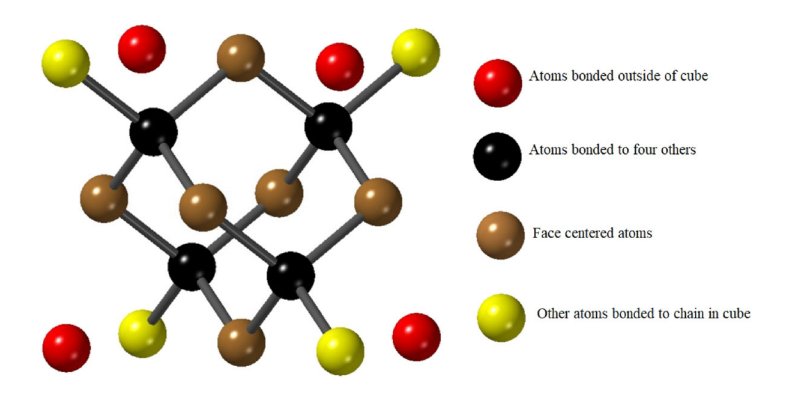}  
  \end{figure} 
  
\begin{center}
  \textbf{Figure 1}. Diamond structure, step by step.  
\end{center}

 \vspace{5mm}
\begin{figure}[h!]
\centering
\includegraphics[width=160mm]{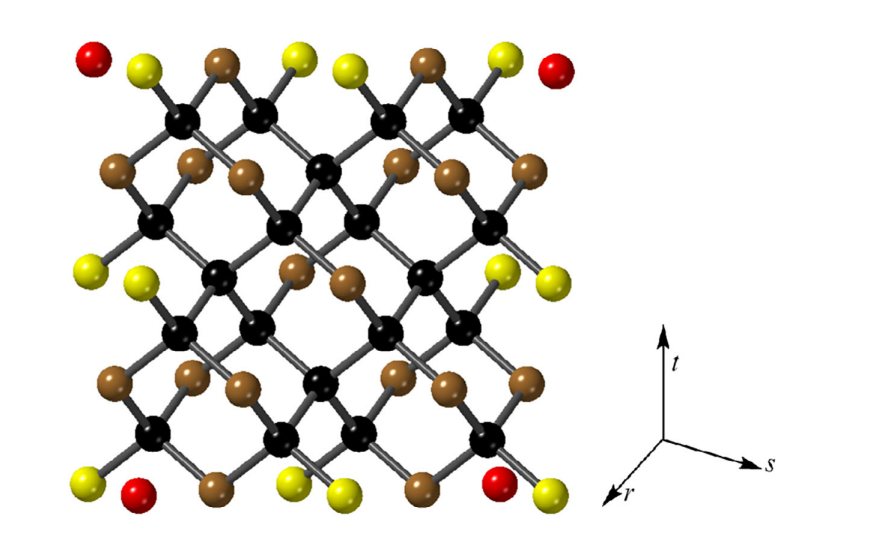}  
  \end{figure} 
\begin{center}
\textbf{Figure 2}. Diamond's tetrahedral structure. 
\end{center}

The authors in \cite{28} computed various degree-based topological indices, such as the first, second, modified, general, and inverse Randi\'c indices, along with the symmetric division index, harmonic index, and inverse sum index for the diamond network, using M-polynomial. They also compared these topological indices with corresponding entropy measures through numerical computations and 2D line plots. 

Motivated by the studies mentioned in \cite{27,28}, we use the M-polynomial to compute the Nirmala index for the diamond structure, along with the first and second inverse Nirmala indices. Moreover, Shannon's entropy model is used to calculate entropy measures for the diamond structure.

\section{Main results}

Throughout the paper, the term ``diamond structure'' refers to the ``tetrahedral structure of diamond'' unless specified otherwise.  

In this section, we first find the M-polynomial of diamond structure. Then, we compute the Nirmala index, and first and second inverse Nirmala indices using its M-polynomial. 

The order and size of the tetrahedral structure of a diamond are $\frac{2s^3+9s^2+13s+6}{6}$ and $\frac{2s^3+6s^2+4s}{3}$, respectively \cite{28}. 

\newpage

The edge set partitions of the diamond structure are given in Table 2.

\vspace{5mm}

\begin{table}[h!]
\centering
\renewcommand{\arraystretch}{2.8}
\begin{tabular}{||c| c |c| c ||} 
 \hline
 Sl. No & Edge set & $(d_{\Upsilon}(u),d_{\Upsilon}(v))$ & Number of repetitions  
 \\ [0.5ex] 
 \hline\hline
 1  & $E_1$  & (1,4) & 4\\
\hline
2  & $E_2$ & (2,4) & $12s-12$ \\
\hline
 3  & $E_3$ & (3,4) & $6s^2-18s+12$  \\
 \hline
 4  & $E_4$ & (4,4) & $\frac{2s^3-12s^2+22s-12}{3}$  \\
\hline 
\end{tabular}
\end{table} 
  
\textbf{Table 2}. Degree-based edge partition of diamond w.r.t end vertices of each edge. \\\\
From Table 2, one can easily observe that 
\begin{center}    
$\displaystyle \sum_{i=1}^{4}|E_{i}|=4+12(s-1)+(6s^2-18s+12)+\left(\frac{2s^3-12s^2+22s-12}{3}\right)=\frac{2s^3+6s^2+4s}{3}$.
\end{center}

\subsection{Nirmala indices of diamond structure}
We now find the M-polynomial of diamond structure as follows.

\begin{theorem}
Let $\Upsilon$ be the molecular graph of the diamond structure. Then the M-polynomial of $\Upsilon$ is \\
$M(\Upsilon; x,y)=4xy^4+(12s-12)x^2y^4+(6s^2-18s+12)x^3y^4+\left(\frac{2s^3-12s^2+22s-12}{3}\right)x^4y^4$.
\end{theorem}
\begin{proof} Let $\Upsilon$ be the molecular graph of diamond structure. From Table 2, 
\begin{center} 
$\displaystyle \sum_{i=1}^{4}|E_{i}|=\frac{2s^3+6s^2+4s}{3}$ 
\end{center}

Since each vertex of $\Upsilon$ is of degree either $1$ or $2$ or $3$ or $4$, 
the partitions of edge set $E(\Upsilon)$ are: \\\\
$E_{1}(\Upsilon):=\{e=uv \in E(\Upsilon): d_{\Upsilon}(u)=1, d_{\Upsilon}(bv)=4 \}$;\\\\
$E_{2}(\Upsilon)):=\{e=uv \in E(\Upsilon): d_{\Upsilon}(u)=2, d_{\Upsilon}(v)=4 \}$;\\\\
$E_{3}(\Upsilon):=\{e=uv \in E(\Upsilon): d_{\Upsilon}(u)=3, d_{\Upsilon}(v)=4 \}$;\\\\
$E_{4}(\Upsilon):=\{e=uv \in E(\Upsilon): d_{\Upsilon}(u)=4, d_{\Upsilon}(v)=4 \}$.\\\\
Clearly, 
\begin{center} 
$|E_{1}(\Upsilon)|=4; |E_{2}(\Upsilon)|=12(s-1); |E_{3}(\Upsilon)|=6(s^2-3s+2); |E_{4}(\Upsilon)|=\frac{2s^3-12s^2+22s-12}{3}$. 
\end{center}
Therefore, 
\begin{align*}
M(\Upsilon; x,y)=&\displaystyle \sum_{\delta \leq i \leq j \leq \Delta} m_{i,j}(\Upsilon)x^{i}y^{j} = m_{14}(\Upsilon)xy^{4}+m_{24} (\Upsilon)x^{2}y^{4}+m_{34}(\Upsilon)x^{3}y^{4} + m_{44}(\Upsilon)x^{4}y^{4} \\ 
 = & 4xy^4+(12s-12)x^2y^4+(6s^2-18s+12)x^3y^4+\left(\frac{2s^3-12s^2+22s-12}{3}\right)x^4y^4.
 \end{align*}
 
 \end{proof}
Now we evaluate the Nirmala indices of diamond structure with the help of its M-polynomial.

\begin{theorem}
Let $\Upsilon$ be the molecular graph for diamond structure. Then the Nirmala indices of $\Upsilon$ are:\\\\
a. $N(\Upsilon)=4\sqrt{5}+\left(12s-12\right)\sqrt{6}+\left(6s^2-18s+12\right)\sqrt{7}+\left(\frac{2s^3-12s^2+22s-12}{3}\right)\sqrt{8}$,\\\\
b. $IN_{1}(\Upsilon)=\sqrt{20}+6\sqrt{3}(s-1)+\sqrt{21}(s^2-3s+2)+\frac{\sqrt{8}}{6}\left(s^3-6s^2+11s-6\right)$,\\\\
c. $IN_{2}(\Upsilon)=\frac{8}{\sqrt{5}}+\frac{24}{\sqrt{3}}(s-1)+\frac{12\sqrt{3}}{\sqrt{7}}(s^2-3s+2)+\frac{\sqrt{8}}{3}\left(s^3-6s^2+11s-6\right)$.
\end{theorem}
\begin{proof} Let $\Upsilon$ be the molecular graph for diamond structure. From Theorem 1, the M-polynomial of $\Upsilon$ is
\begin{center}
$M(\Upsilon; x,y)=4xy^4+(12s-12)x^2y^4+(6s^2-18s+12)x^3y^4+\left(\frac{2s^3-12s^2+22s-12}{3}\right)x^4y^4$.
\end{center} 
From Table 1, we have 
\begin{align*}
& (i) \, D_{x}^{\frac{1}{2}}J(M(\Upsilon; x,y)) \\
& = D_{x}^{\frac{1}{2}}J \left[4xy^4+(12s-12)x^2y^4+(6s^2-18s+12)x^3y^4+\left(\frac{2s^3-12s^2+22s-12}{3}\right)x^4y^4 \right] \\
& = D_{x}^{\frac{1}{2}} \left[4x^5+(12s-12)x^6+(6s^2-18s+12)x^7+\left(\frac{2s^3-12s^2+22s-12}{3}\right)x^8 \right] \\
& = 4\sqrt{5}x^5+(12s-12)\sqrt{6}x^6+(6s^2-18s+12)\sqrt{7}x^7+\left(\frac{2s^3-12s^2+22s-12}{3}\right)\sqrt{8}x^8.
\end{align*}
\begin{align*}
& (ii) \, D_{x}^{\frac{1}{2}}JS_{y}^{\frac{1}{2}}S_{x}^{\frac{1}{2}}(M(\Upsilon; x,y)) \\
=&D_{x}^{\frac{1}{2}}JS_{y}^{\frac{1}{2}}S_{x}^{\frac{1}{2}} \left[4xy^4+(12s-12)x^2y^4+(6s^2-18s+12)x^3y^4+\left(\frac{2s^3-12s^2+22s-12}{3}\right)x^4y^4 \right] \\
= & D_{x}^{\frac{1}{2}}JS_{y}^{\frac{1}{2}} \left[4xy^4+\frac{1}{\sqrt{2}}(12s-12)x^2y^4+\frac{1}{\sqrt{3}}(6s^2-18s+12)x^3y^4+\frac{1}{2}\left(\frac{2s^3-12s^2+22s-12}{3}\right)x^4y^4 \right] \\
& = D_{x}^{\frac{1}{2}}J \left[\sqrt{4}xy^4+\frac{6}{\sqrt{2}}
(s-1)x^2y^4+\frac{1}{2\sqrt{3}}(6s^2-18s+12)x^3y^4+\frac{1}{4}\left(\frac{2s^3-12s^2+22s-12}{3}\right)x^4y^4 \right] \\
& = D_{x}^{\frac{1}{2}}\left[\sqrt{4}x^5+\frac{6}{\sqrt{2}}(s-1)x^6+\frac{1}{2\sqrt{3}}(6s^2-18s+12)x^7+\frac{1}{4}\left(\frac{2s^3-12s^2+22s-12}{3}\right)x^8 \right] \\
& = D_{x}^{\frac{1}{2}}\left[\sqrt{4}x^5+\frac{6}{\sqrt{2}}(s-1)x^6+\sqrt{3}(s^2-3s+2)x^7+\frac{1}{6}\left(s^3-6s^2+11s-6\right)x^8 \right] \\
& = \sqrt{20}x^5+6\sqrt{3}(s-1)x^6+\sqrt{21}(s^2-3s+2)x^7+\frac{\sqrt{8}}{6}\left(s^3-6s^2+11s-6\right)x^8.
\end{align*}

\begin{align*}
(iii) \, & S_{x}^{\frac{1}{2}}JD_{y}^{\frac{1}{2}}D_{x}^{\frac{1}{2}}(M(\Upsilon; x,y))  \\
&=S_{x}^{\frac{1}{2}}JD_{y}^{\frac{1}{2}}D_{x}^{\frac{1}{2}} \left[ 4xy^4+(12s-12)x^2y^4+(6s^2-18s+12)x^3y^4+\left(\frac{2s^3-12s^2+22s-12}{3}\right)x^4y^4\right] \\
&=S_{x}^{\frac{1}{2}}JD_{y}^{\frac{1}{2}} \left[4xy^4+12\sqrt{2}(s-1)x^2y^4+6\sqrt{3}(s^2-3s+2)x^3y^4+2\left(\frac{2s^3-12s^2+22s-12}{3}\right)x^4y^4\right] \\
&=S_{x}^{\frac{1}{2}}J \left[8xy^4+24\sqrt{2}(s-1)x^2y^4+12\sqrt{3}(s^2-3s+2)x^3y^4+4\left(\frac{2s^3-12s^2+22s-12}{3}\right)x^4y^4 \right] \\
&=S_{x}^{\frac{1}{2}} \left[8x^5+24\sqrt{2}(s-1)x^6+12\sqrt{3}(s^2-3s+2)x^7+4\left(\frac{2s^3-12s^2+22s-12}{3}\right)x^8\right] \\
&= \frac{8}{\sqrt{5}}x^5+\frac{24\sqrt{2}}{\sqrt{6}}(s-1)x^6+\frac{12\sqrt{3}}{\sqrt{7}}(s^2-3s+2)x^7+\frac{4}{3\sqrt{8}}\left(2s^3-12s^2+22s-12\right)x^8\\
& =\frac{8}{\sqrt{5}}x^5+\frac{24}{\sqrt{3}}(s-1)x^6+\frac{12\sqrt{3}}{\sqrt{7}}(s^2-3s+2)x^7+\frac{\sqrt{8}}{3}\left(s^3-6s^2+11s-6\right)x^8. 
\end{align*}

Hence, the Nirmala indices of $\Upsilon$ are given by 

\begin{align*}
(a) \, N(\Upsilon)=& D_{x}^{\frac{1}{2}}J(M(\Upsilon; x,y))|_{x=1} \\
& = 4\sqrt{5}+\left(12s-12\right)\sqrt{6}+\left(6s^2-18s+12\right)\sqrt{7}+\left(\frac{2s^3-12s^2+22s-12}{3}\right)\sqrt{8}.
\end{align*} 
\begin{align*}
(b) \, IN_{1}(\Upsilon)=&D_{x}^{\frac{1}{2}}JS_{y}^{\frac{1}{2}}S_{x}^{\frac{1}{2}}(M(\Upsilon; x,y))|_{x=1}   \\
 & = \sqrt{20}+6\sqrt{3}(s-1)+\sqrt{21}(s^2-3s+2)+\frac{\sqrt{8}}{6}\left(s^3-6s^2+11s-6\right).
 \end{align*}
 \begin{align*}     
(c) \, IN_{2}(\Upsilon)=& S_{x}^{\frac{1}{2}}JD_{y}^{\frac{1}{2}}D_{x}^{\frac{1}{2}}(M(\Upsilon; x,y))|_{x=1}   \\
&= \frac{8}{\sqrt{5}}+\frac{24}{\sqrt{3}}(s-1)+\frac{12\sqrt{3}}{\sqrt{7}}(s^2-3s+2)+\frac{\sqrt{8}}{3}\left(s^3-6s^2+11s-6\right).
\end{align*} 
\end{proof}

\subsection{Entropy measures of diamond structure }
Using Shannon's entropy, we continue to compute graph entropy measures for the diamond structure. We use the previously determined expressions of the Nirmala indices to derive the mathematical formulations for entropy measures based on Nirmala indices.\\\\
\textbf{Nirmala entropy of diamond structure:}\\
From Theorem 2, the  Nirmala index of $\Upsilon$ is given by
\begin{equation*}
N(\Upsilon)=4\sqrt{5}+\left(12s-12\right)\sqrt{6}+\left(6s^2-18s+12\right)\sqrt{7}+\left(\frac{2s^3-12s^2+22s-12}{3}\right)\sqrt{8}.     
\end{equation*}

From Table 2 and Equation (5), we have
\begin{align*}
 ENT_{N}(\Upsilon)=&log(N(\Upsilon))  - \frac{1}{N(\Upsilon)}              \displaystyle \sum_{uv \in E(\Upsilon)} \sqrt{d_{\Upsilon}(u)+d_{\Upsilon}(v)} \times log (\sqrt{d_{\Upsilon}(u)+d_{\Upsilon}(v)}) \\
 = & log(N(\Upsilon))-\frac{1}{N(\Upsilon)} \left[\displaystyle \sum_{i=1}^{4} \displaystyle \sum_{uv \in E_{i}(\Upsilon)} \sqrt{d_{\Upsilon}(u)+d_{\Upsilon}(v)} \times log (\sqrt{d_{\Upsilon}(u)+d_{\Upsilon}(v)}) \right] \\
 & = log(N(\Upsilon))-\frac{1}{N(\Upsilon)} \left[4 \cdot \sqrt{5} \cdot log(\sqrt{5}) + 12(s-1) \cdot \sqrt{6} \cdot log(\sqrt{6}) \right] \\
 & -\frac{1}{N(\Upsilon)} \left[6(s^2-3s+2) \cdot \sqrt{7} \cdot log(\sqrt{7}) + \left(\frac{2s^3-12s^2+22s-12}{3}\right) \cdot \sqrt{8} \cdot log(\sqrt{8}) \right]  
  \end{align*}
Finally, we get the desired formulation of the Nirmala entropy for diamond structure by substituting the value of $N(\Upsilon)$ into the previous expression.\\\\
\textbf{First inverse Nirmala entropy of diamond structure:}\\
From Theorem 2, the first inverse Nirmala index of $\Upsilon$ is given by
\begin{align*}
IN_{1}(\Upsilon)=& \sqrt{20}+6\sqrt{3}(s-1)+\sqrt{21}(s^2-3s+2)+\frac{\sqrt{8}}{6}\left(s^3-6s^2+11s-6\right).    
\end{align*}
From Table 2 and Equation (6), we have
\begin{align*}
 ENT_{IN_{1}}(\Upsilon)=&log(IN_{1}(\Upsilon))  - \frac{1}{IN_{1}(\Upsilon)}              \displaystyle \sum_{uv \in E(\Upsilon)} \sqrt{\frac{1}{d_{\Upsilon}(u)}+\frac{1}{d_{\Upsilon}(v)}} \times log \left( \sqrt{\frac{1}{d_{\Upsilon}(u)}+\frac{1}{d_{\Upsilon}(v)}} \right) \\
 =& log(IN_{1}(\Upsilon)) - \frac{1}{IN_{1}(\Upsilon)} \left[\displaystyle \sum_{i=1}^{4} \displaystyle \sum_{uv \in E_{i}(\Upsilon)} \sqrt{\frac{1}{d_{\Upsilon}(u)}+\frac{1}{d_{\Upsilon}(v)}} \times log \left( \sqrt{\frac{1}{d_{\Upsilon}(u)}+\frac{1}{d_{\Upsilon}(v)}} \right) \right] \\
 = & log(IN_{1}(\Upsilon)) - \frac{1}{IN_{1}(\Upsilon)} \left[4 \cdot \sqrt{\frac{1}{4}+\frac{1}{4}} \cdot log\left(\sqrt{\frac{1}{4}+\frac{1}{4}}\right) \right] \\
 - & \frac{1}{IN_{1}(\Upsilon)} \left[12(s-1) \cdot \sqrt{\frac{2}{4}+\frac{2}{4}} \cdot log\left(\sqrt{\frac{2}{4}+\frac{2}{4}}\right) \right] \\
 - & \frac{1}{IN_{1}(\Upsilon)} \left[6(s^2-3s+2) \cdot \sqrt{\frac{3}{4}+\frac{3}{4}} \cdot log\left(\sqrt{\frac{3}{4}+\frac{3}{4}}\right) \right] \\
 - & \frac{1}{IN_{1}(\Upsilon)} \left[ \left(\frac{2s^3-12s^2+22s-12}{3}\right)  \cdot \sqrt{\frac{4}{4}+\frac{4}{4}} \cdot log\left(\sqrt{\frac{4}{4}+\frac{4}{4}}\right) \right] \\ 
 \end{align*}
 Since $log(1)$=0, 
 \begin{align*}
  ENT_{IN_{1}}(\Upsilon) & =log(IN_{1}(\Upsilon)) - \frac{1}{IN_{1}(\Upsilon)} \left[4 \cdot \sqrt{\frac{1}{2}} \cdot log\left(\sqrt{\frac{1}{2}}\right) \right] \\ 
 - & \frac{1}{IN_{1}(\Upsilon)} \left[6(s^2-3s+2) \cdot \sqrt{\frac{3}{2}} \cdot log\left(\sqrt{\frac{3}{2}}\right) + \left(\frac{2s^3-12s^2+22s-12}{3}\right)  \cdot \sqrt{2} \cdot log\left(\sqrt{2}\right) \right]  
 \end{align*}

Finally, by substituting the value of $IN_{1}(\Upsilon)$ into the preceding expression, we obtain the desired formulation of the first inverse Nirmala entropy for diamond structure.\\\\
\textbf{Second inverse Nirmala entropy of diamond structure:}\\
From Theorem 2, the second inverse Nirmala index of $\Upsilon$ is given by
\begin{align*}
IN_{2}(\Upsilon)= & \frac{8}{\sqrt{5}}+\frac{24}{\sqrt{3}}(s-1)+\frac{12\sqrt{3}}{\sqrt{7}}(s^2-3s+2)+\frac{\sqrt{8}}{3}\left(s^3-6s^2+11s-6\right).
\end{align*} 
From Table 2 and Equation (7), we have
\begin{align*}
 ENT_{IN_{2}}(\Upsilon) =  &  log(IN_{2}(\Upsilon))  - \frac{1}{IN_{2}(\Upsilon)}              \displaystyle \sum_{uv \in E(\Upsilon)} \frac{\sqrt{d_{\Upsilon}(u) \cdot  d_{\Upsilon}(v)} }{\sqrt{d_{\Upsilon}(u) + d_{\Upsilon}(v)} } \times log \left( \frac{\sqrt{d_{\Upsilon}(u) \cdot  d_{\Upsilon}(v)} }{\sqrt{d_{\Upsilon}(u) + d_{\Upsilon}(v)} } \right) \\
 = & log(IN_{2}(\Upsilon))  - \frac{1}{IN_{2}(\Upsilon)} \left[ \displaystyle \sum_{i=1}^{4} \displaystyle \sum_{uv \in E_{i}(\Upsilon)}              \frac{\sqrt{d_{\Upsilon}(u) \cdot  d_{\Upsilon}(v)} }{\sqrt{d_{\Upsilon}(u) + d_{\Upsilon}(v)} } \times log \left( \frac{\sqrt{d_{\Upsilon}(u) \cdot  d_{\Upsilon}(v)} }{\sqrt{d_{\Upsilon}(u) + d_{\Upsilon}(v)} } \right) \right] \\
 = & log(IN_{2}(\Upsilon))  - \frac{1}{IN_{2}(\Upsilon)} \left[ 4 \cdot \frac{\sqrt{4}}{\sqrt{5}} \cdot log  \left(\frac{\sqrt{4}}{\sqrt{5}} \right) + 12(s-1) \cdot \frac{\sqrt{8}}{\sqrt{6}} \cdot log  \left(\frac{\sqrt{8}}{\sqrt{6}} \right) \right]\\  
 - & \frac{1}{IN_{2}(\Upsilon)} \left[6(s^2-3s+2) \cdot \frac{\sqrt{12}}{\sqrt{7}} \cdot log \left(\frac{\sqrt{12}}{\sqrt{7}} \right) + \frac{2s^3-12s^2+22s-12}{3} \cdot \frac{\sqrt{16}}{\sqrt{8}} \cdot log \left(\frac{\sqrt{16}}{\sqrt{8}} \right) \right] \\
& = log(IN_{2}(\Upsilon))  - \frac{1}{IN_{2}(\Upsilon)} \left[  \frac{8}{\sqrt{5}} \cdot log  \left(\frac{2}{\sqrt{5}} \right) + 12(s-1) \cdot \frac{\sqrt{8}}{\sqrt{6}} \cdot log  \left(\frac{\sqrt{8}}{\sqrt{6}} \right) \right]\\  
 - & \frac{1}{IN_{2}(\Upsilon)} \left[6(s^2-3s+2) \cdot \frac{\sqrt{12}}{\sqrt{7}} \cdot log \left(\frac{\sqrt{12}}{\sqrt{7}} \right) + \left(\frac{2s^3-12s^2+22s-12}{3}\right) \cdot \sqrt{2} \cdot log(\sqrt{2})  \right]
\end{align*}
The second inverse Nirmala entropy for the diamond structure can finally be expressed as desired by substituting the value of $IN_{2}(\Upsilon)$ in the previous expression.

\section{Comparison through numerical and graphical demonstrations}
 Graph entropy metrics are widely used in many scientific domains, including computer science, information theory, chemistry, biological therapies, and pharmacology. Therefore, scientists working in these disciplines rely on numerical calculation and graphical representation to appropriately characterize these molecular characteristics. This section compares the Nirmala indices and the corresponding entropy measures using numerical computation and 2D line graphs. 
 
 Table 3 shows the results of the numerical calculation of the Nirmala indices and corresponding entropy measures for the diamond structure. Table 3 contains the intervals $1\leq s \leq 25$. Additionally, Figure 3 uses 2D line plots for $1\leq s \leq 25$ to compare the Nirmala indices and the corresponding entropy measures. 
\newpage

\begin{center} 
\begin{table}[h!]
\centering
\renewcommand{\arraystretch}{1.6}
\begin{tabular}{|>{\centering\arraybackslash}m{0.8cm}||>{\centering\arraybackslash}m{2.2cm} |>{\centering\arraybackslash}m{2.2cm} |>{\centering\arraybackslash}m{2.2cm} |>{\centering\arraybackslash}m{2cm} |>{\centering\arraybackslash}m{2cm} |>{\centering\arraybackslash}m{2cm} ||} 
 \hline
 [s] & N & $IN_{1}$ & $IN_{2}$ & $ENT_{N}$ &  $ENT_{IN_{1}}$  &   $ENT_{IN_{2}}$
 \\ [0.5ex] 
 \hline\hline
[1] & 8.94 & 4.47 & 3.57 & 1.38 & 1.71 & 1.38 \\
\hline
[2] & 38.33 & 14.46 & 17.43 & 2.77 & 2.76 & 2.76  \\ 
\hline
[3] &  99.48 & 34.42 & 47.00  & 3.68 & 3.48 & 3.68\\ 
 \hline 
[4] & 203.68 & 68.80 & 97.93 & 4.38 & 4.08 & 4.37\\ 
\hline 
[5] & 362.26 & 123.65 & 175.90 & 4.94 & 4.61 & 4.93 \\
\hline 
[6] & 586.54 & 204.65 & 286.54 & 5.41 & 5.08 & 5.40 \\
\hline
[7] & 887.81  & 317.44 & 435.52 & 5.81 & 5.49  & 5.81 \\
\hline
[8] & 1277.41 & 467.67 & 628.50 & 6.17 & 5.86 & 6.17 \\
\hline 
[9] & 1766.63 & 661.01 & 871.14 & 6.49 & 6.20 & 6.48\\
\hline  
[10] & 2366.80 & 903.12 & 1169.08 & 6.77 & 6.50 & 6.77\\
\hline
[11] & 3089.23 & 1199.64 & 1527.99 & 7.04 & 6.78 & 7.04 \\
\hline 
[12] & 3945.23 & 1556.25 & 1953.52 & 7.28 & 7.03 & 7.28 \\
\hline 
[13] & 4946.12 & 1978.58 & 2451.33 & 7.50 & 7.27 & 7.50 \\
\hline 
[14] & 6103.20  & 2472.31 & 3027.08  & 7.71 & 7.49 & 7.71 \\
\hline
[15] & 7427.80 & 3043.08 & 3686.42 & 7.90 & 7.69 & 7.90 \\
\hline
[16] & 8931.23 & 3696.56 & 4435.01 & 8.08 & 7.88 & 8.08 \\
\hline 
[17] & 10624.80 & 4438.40 & 5278.52 & 8.26 & 8.07 & 8.26 \\
\hline 
[18] & 12519.82 & 5274.26  & 6222.58  & 8.42 & 8.24 & 8.42 \\
\hline 
[19] & 14627.61  & 6209.79 & 7272.87 & 8.57 & 8.40 & 8.57 \\
\hline 
[20] & 16959.49 & 7250.65 & 8435.03 & 8.72 & 8.55 & 8.72 \\
\hline
[21] & 19526.76 & 8402.51 & 9714.74 & 8.86 & 8.70  & 8.86\\
\hline
[22] & 22340.74 & 9671.00  & 11117.63 & 8.99 & 8.84 & 8.99 \\
\hline
[23] & 25412.74 & 11061.80 & 12649.37 & 9.12 & 8.97 & 9.12  \\
\hline
[24] & 28754.08 & 12580.56 & 14315.62 & 9.24 & 9.10 & 9.24\\
\hline
[25] & 32376.07 & 14232.94 & 16122.03 & 9.36 & 9.22 & 9.36 \\
\hline
\end{tabular}
\end{table} 
\end{center}  
\textbf{Table 3}. The results of the numerical calculation of the Nirmala indices and corresponding entropy measures for the diamond structure, where $1 \leq s \leq 25$.

\newpage

 \vspace{5mm}
\begin{figure}[h!]
\centering
\includegraphics[width=145mm]{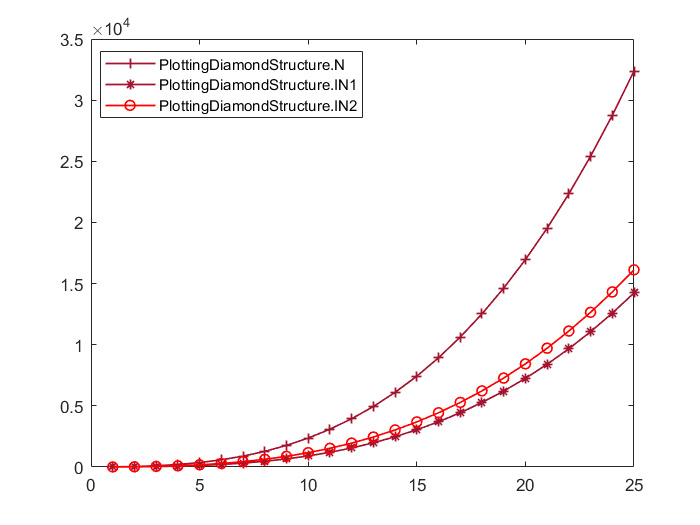}  
  \end{figure} 

\begin{figure}[h!]
\centering
\includegraphics[width=145mm]{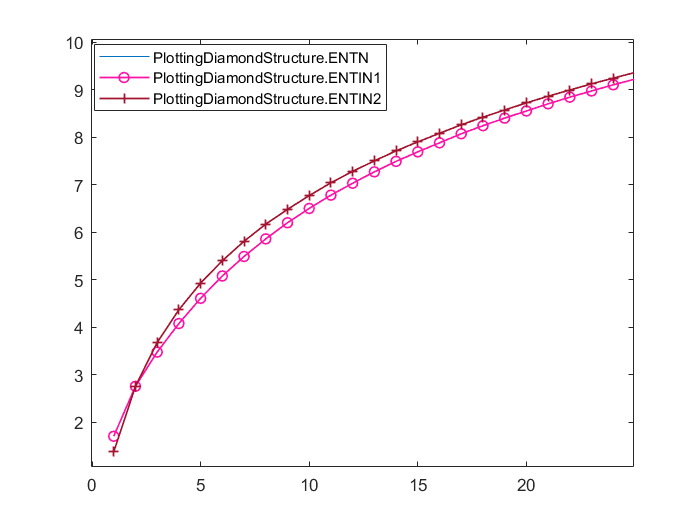}  
  \end{figure} 
\begin{center}
\textbf{Figure 3}. Comparison of the Nirmala indices and their associated entropy measures of diamond structure through a 2D line plot for $1 \leq s \leq 25$. 
\end{center}
\newpage
From Table 3 and Figure 3, the following two remarks are possible. \\
\textbf{Remark 1:}  The Nirmala indices and associated entropy measures of the diamond structure increase as the values of $s$ increase. \\
\textbf{Remark 2:} For the molecular graph $\Upsilon$ diamond structure, we have the following inequality relationships: \\\\
(i) $IN_2(\Upsilon) < IN_1(\Upsilon) < N(\Upsilon)$ for $s=1$; \\\\ 
(ii) $IN_1(\Upsilon) < IN_2(\Upsilon) < N(\Upsilon)$ for any $2 \leq s \leq 25$. \\\\
(iii) $ENT_{N}(\Upsilon) \approx ENT_{IN_{2}}(\Upsilon)$ for any $1 \leq s \leq 25$; \\\\
(iv) $ENT_{N}(\Upsilon) < ENT_{IN_{1}}(\Upsilon)$ and $ENT_{IN_{2}}(\Upsilon) <  ENT_{IN_{1}}(\Upsilon)$ for $s=1$; \\\\
(v) $ENT_{IN_{1}}(\Upsilon) < ENT_{N}(\Upsilon)$ and $ENT_{IN_{1}}(\Upsilon) <  ENT_{IN_{2}}(\Upsilon)$ for any $2 \leq s \leq 25$.

\subsection{The logarithmic regression model}
We use logarithmic regression analysis to look into the relationship between the dependent variable and one or more predictor variables in our dataset. Logarithmic regression is a nonlinear technique that modifies the dependent or predictor variables. The logarithmic regression model equation is given by 
\begin{equation*}
y=a*log(x)+b, 
\end{equation*} 
where the response variable is $y$ and the predictor variable is $x$. The regression coefficients $a$ and $b$ represent the relation between $x$ and $y$.

Here, for $1 \leq s \leq 25$, we examine the relationship between the Nirmala indices and entropy metrics of the diamond structure using a logarithmic regression analysis. The study utilized many statistical measurements, such as the squared correlation coefficient ($R^{2}$), the sum of square error (SSE), adjusted squared correlation coefficient (Adj. R-sq), root mean square error (RMSE), and squared correlation coefficient ($R^{2}$). A low RMSE value (nearer to 0) implies that the model performs well, whereas a larger $R^{2}$ value (near 1) suggests that the regression line fits the data better. In this instance, obtaining a larger $R^{2}$ value is our main goal. 

The statistics of curve fitting of the Nirmala indices versus Nirmala entropy measures for diamond structure using the logarithmic regression are shown in Table 4.
\newpage
\vspace{5mm}
\begin{table}[h!]
\centering
\renewcommand{\arraystretch}{3}
\begin{tabular}{|>{\centering\arraybackslash}m{7cm}||>{\centering\arraybackslash}m{1.8cm} |>{\centering\arraybackslash}m{1.8cm} |>{\centering\arraybackslash}m{2cm} |>{\centering\arraybackslash}m{2cm} ||} 
 
 \hline
 Model & $R^{2}$ & SSE & Adj. R-sq & RMSE
 \\ [0.5ex] 
 \hline\hline
 $ENT_{N}=0.9772*log(N)-0.8050$  &  \textbf{0.9999}   & 0.0048 & 0.9999 & 0.0144 \\
\hline
 $ENT_{IN_{1}}=0.9429*log(IN_{1})+0.1395$  & \textbf{0.9997}   & 0.0892 & 0.9991 & 0.0622  \\
\hline
 $ENT_{IN_{2}}=0.9614*log(IN_{2})+0.0083$  & \textbf{0.9996}    & 0.0406 & 0.9996  & 0.0420   \\
 \hline
\end{tabular}
\end{table} 
  
\textbf{Table 4}. Statistics of curve fitting of the Nirmala indices vs. Nirmala entropy measures of the diamond structure.

\vspace{15mm}
\begin{figure}[h!]
\centering
\includegraphics[width=150mm]{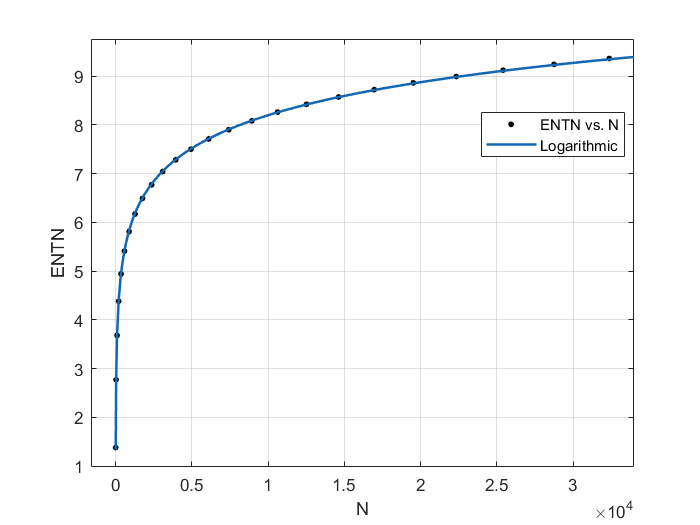}  
  \end{figure} 
 \newpage
  \begin{figure}[h!]
\centering
\includegraphics[width=148mm]{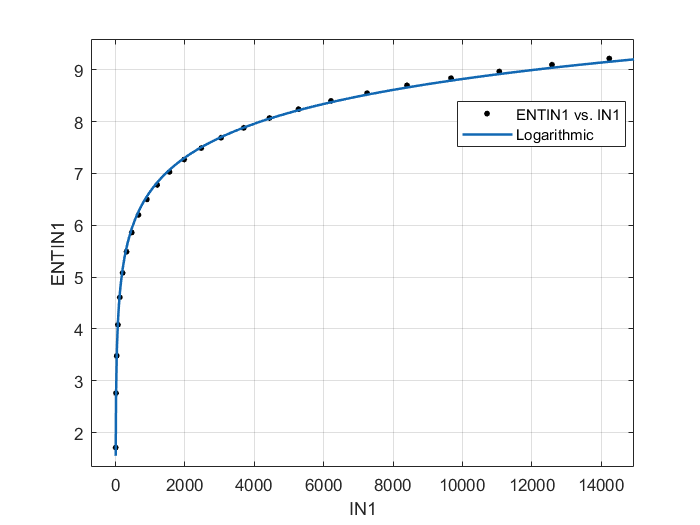}  
  \end{figure} 
 \begin{figure}[h!]
\centering
\includegraphics[width=148mm]{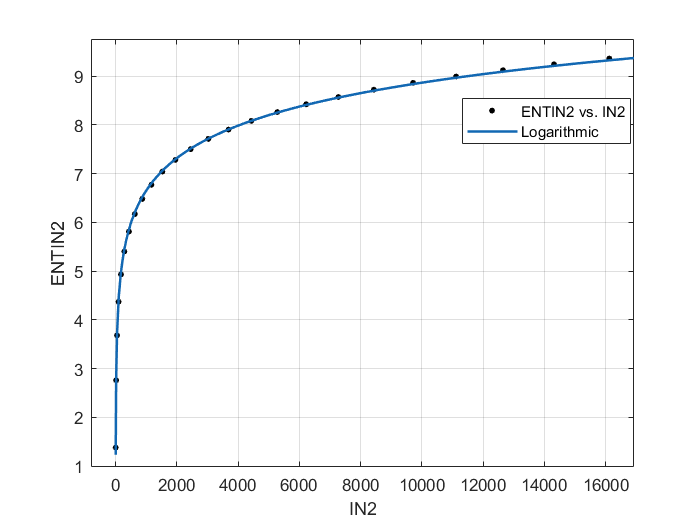}  
  \end{figure} 
  \textbf{Figure 3}. Curve fitting plots for the Nirmala indices vs. Nirmala entropy measures of the diamond structure.
  
 \section{Discussion} 
Topological indices are numerical representations of molecular groups in graph theory that help researchers better understand the molecular features and behavior of chemical and biological systems. Accurate calculation of topological indices gives researchers relevant information that enhances their knowledge of the topic. In this work, we study the so-called degree-based topological indices, or Nirmala indices, of the diamond structure. Table 3 shows that the Nirmala indices and associated entropy measures of the diamond structure increase together with an increase in $s$. Entropy metrics assess a dataset's uncertainty or information content to aid in the estimation of its complexity and distribution. These ideas are fundamental to information theory, thermodynamics, and data analysis. Accurate numerical entropy calculations provide researchers with important new information that enhances their comprehension of the network under study. Given its benefits, it supports our attention to the diamond structure's edge weight entropy. 

The social sciences, biology, and economics use the non-linear regression approach known as logarithmic regression to make predictions and account for non-linearity in data. Table 4 shows the statistical measures used in the study, including the sum of square errors (SSE), adjusted squared correlation coefficient (Adj. R-sq), squared correlation coefficient ($R^{2}$), and root mean square error (RMSE). A higher $R^{2}$ value, which is closer to 1, indicates a regression line that fits the data more accurately. Figure 3 illustrates how well the Nirmala indices and the associated entropy measure values of the diamond structure seem to fit the curve.
\section{Conclusion}
The definitions of the Nirmala indices and entropy measures based on the Nirmala indices have been used in this research work. The Nirmala indices of the diamond structure have been mathematically formulated. Its M-polynomial has been used to analyze the Nirmala indices-based entropy measures of this structure. Additionally, the Nirmala indices and related entropy measures have been numerically computed and compared them using 2D line plots. The Nirmala indices and corresponding entropy measures of the diamond structure increase as $s$ increases, according to Table 3 and Figure 3. 

For the molecular graph $\Upsilon$ diamond structure, we have the following inequality relationships: \\\\
(i) $IN_2(\Upsilon) < IN_1(\Upsilon) < N(\Upsilon)$ for $s=1$; \\\\ 
(ii) $IN_1(\Upsilon) < IN_2(\Upsilon) < N(\Upsilon)$ for any $2 \leq s \leq 25$. \\\\
(iii) $ENT_{N}(\Upsilon) \approx ENT_{IN_{2}}(\Upsilon)$ for any $1 \leq s \leq 25$; \\\\
(iv) $ENT_{N}(\Upsilon) < ENT_{IN_{1}}(\Upsilon)$ and $ENT_{IN_{2}}(\Upsilon) <  ENT_{IN_{1}}(\Upsilon)$ for $s=1$; \\\\
(v) $ENT_{IN_{1}}(\Upsilon) < ENT_{N}(\Upsilon)$ and $ENT_{IN_{1}}(\Upsilon) <  ENT_{IN_{2}}(\Upsilon)$ for any $2 \leq s \leq 25$. \\\\
The findings of this study will be useful in examining the topology and structural properties of the diamond structure in the domains of electronic, mechanical, optical, and nanoelectronic technology.

\section*{Funding} No funding is available for this study.
\section*{Author contributions} All authors contributed equally.
\section*{Data Availability Statement}
This manuscript has no associated data.
\section*{Declarations}
\textbf{Conflict of interest} The authors declare that they have no known competing financial interests or personal relationships that could have appeared to influence the work reported in this paper.
  \makeatletter
\renewcommand{\@biblabel}[1]{[#1]\hfill}
\newpage
\makeatother

\end{document}